\documentclass{llncs}

\usepackage[T1]{fontenc}

\usepackage{graphicx}
\usepackage{color}
\usepackage{hyperref}

\usepackage{amsmath}
\usepackage{amsfonts}
\usepackage{enumitem}
\usepackage{fdsymbol}
\usepackage{caption}
\usepackage{subcaption}

\newcommand{\ignore}[1]{}

\newcommand{\ADS}{\emph{ADS}}
\newcommand{\Runtime}{\emph{Runtime}}

\newcommand{\eqdef}{\stackrel{\mathit{def}}{=}}

\newcommand{\Reals}{\mathbb{R}}
\newcommand{\PosReals}{\Reals_{\ge 0}}

\newcommand{\Segments}{{\cal S}}
\newcommand{\segnorm}[1]{|\!| #1 |\!|}
\newcommand{\subseg}[3]{{#1}{[#2,#3]}}

\newcommand{\pre}[1]{{^\bullet\!{#1}}}
\newcommand{\post}[1]{{{#1}^\bullet}}

\newcommand{\segment}[1]{{#1}.s}
\newcommand{\Pos}[1]{\mathit{Pos}_{#1}}
\newcommand{\edgeride}[4]{{#2} \xrightarrow{{#3}}_{#1} {#4}}
\newcommand{\ride}[4]{{#2} \stackrel{#3}\rightsquigarrow_{#1} {#4}}
\newcommand{\distance}[1]{d_{#1}}

\newcommand{\priority}[1]{\prec_{#1}}

\newcommand{\att}[2]{#1.\mathit{#2}}
\newcommand{\po}[1]{\att{#1}{p}}
\newcommand{\dt}[1]{\att{#1}{\delta}}
\newcommand{\speed}[1]{\att{#1}{v}}
\newcommand{\itn}[1]{\att{#1}{it}}
\newcommand{\fs}[1]{\att{#1}{f}}
\newcommand{\limit}[1]{\att{#1}{\pi}}
\newcommand{\wt}[1]{\att{#1}{wt}}

\newcommand{\Vehicles}{{\cal C}}
\newcommand{\Objects}{{\cal O}}

\newcommand{\Brake}{B}
\newcommand{\Ahead}[2]{\mathit{Ahead}_{#1}(#2)}
\newcommand{\ahead}[2]{\mathit{ahead}_{#1}(#2)}

\newcommand{\edgemeets}[2]{\edgeride{}{{#1}}{}{{#2}} }
\newcommand{\meets}[2]{\ride{}{{#1}}{}{{#2}}}
\newcommand{\at}[2]{ {#1} \mathit{@} {#2} }

\newcommand{\capacity}[1]{\mathit{w}({#1})}

\newcommand{\always}{\Box}

\newcommand{\next}{\mathsf{N}~}

\newcommand{\cl}[1]{\swarrow\!\!\!{#1}}

\begin{document}

\title{Correct by Design Coordination of Autonomous Driving Systems}

\author{Marius Bozga\orcidID{0000-0003-4412-5684} \and
Joseph Sifakis\orcidID{0000-0003-2447-7981}}

\authorrunning{M. Bozga, J. Sifakis}

\institute{Univ. Grenoble Alpes, CNRS, Grenoble
  INP\footnote{Institute of Engineering Univ. Grenoble Alpes},
  VERIMAG, 38000 Grenoble, France
  \email{\{Marius.Bozga,Joseph.Sifakis\}@univ-grenoble-alpes.fr}\\
  \url{http://www-verimag.imag.fr/}}

\maketitle

\begin{abstract}
The paper proposes a method for the correct by design coordination of
autonomous driving systems (\ADS). It builds on previous results on
collision avoidance policies and the modeling of \ADS\ by combining
descriptions of their static environment in the form of maps, and the
dynamic behavior of their vehicles.
  
An \ADS\ is modeled as a dynamic system involving a set of vehicles
coordinated by a \Runtime\ that based on vehicle positions on a map and
their kinetic attributes, computes free spaces for each vehicle. Vehicles
are bounded to move within the corresponding allocated free spaces.
  
We provide a correct by design safe control policy for an \ADS\ if its
vehicles and the \Runtime\ respect corresponding assume-guarantee
contracts. The result is established by showing that the composition of
assume-guarantee contracts is an inductive invariant that entails
\ADS\ safety.
  
We show that it is practically possible to define speed control policies
for vehicles that comply with their contracts.  Furthermore, we show that
traffic rules can be specified in a linear-time temporal logic, as a class
of formulas that constrain vehicle speeds. The main result is that, given a
set of traffic rules, it is possible to derive free space policies of the
\Runtime\ such that the resulting system behavior is safe by design with
respect to the rules.

\keywords{Autonomous driving systems \and Traffic rule specification \and
  Map specification \and Collision avoidance policy \and Assume-guarantee
  contract \and Correctness by design.}
\end{abstract}

\section{Introduction}
Autonomous driving systems (\ADS) are probably the most difficult
systems to design and validate, because the behavior of their agents
is subject to temporal and spatial dynamism. They are real-time
distributed systems involving components with partial knowledge of
their environment, pursuing specific goals while the collective
behavior must meet given global goals.

Development of trustworthy \ADS\ is an urgent and critical need. It
poses challenges that go well beyond the current state of the art due
to their overwhelming complexity.  These challenges include, on the
one hand, modeling the system and specifying its properties, usually
expressed as traffic rules; on the other hand, building the system and
verifying its correctness with respect to the desired system
properties.

Modeling involves a variety of issues related to the inherent temporal and
spatial dynamics as well as to the need for an accurate representation of
the physical environment in which vehicles operate. Much work focuses on
formalizing and standardizing a concept of map that is central to semantic
awareness and decision-making. Maps should be semantically rich, i.e.,
their structure should provide all relevant semantic information such as
road type, junction type, signaling information, etc. Given their
importance for the modeling \ADS, maps have been the subject of many
studies focused on their compositional description and formalization. These
studies often use ontologies and logics with associated reasoning
mechanisms to check the consistency of descriptions and their accuracy with
respect to desired properties \cite{BeetzB18,BagschikMM18}. Other works
propose open source mapping frameworks for highly automated
driving \cite{OpenDRIVE-1.4,ASAMOpenDRIVE-1.6.0,PoggenhansPJONK18}.
Finally, the SOCA method \cite{ButzHHORSZ20} proposes an
abstraction of maps called zone graph, and uses this abstraction in a
morphological behavior analysis.

There is an extensive literature on \ADS\ validation that involves two
interrelated problems: the specification of system properties and the
application of validation techniques.  The specification of properties
requires first-order temporal logics because parameterization and
genericity are essential for the description of situations involving a
varying number of vehicles and types of traffic patterns.  The work
in \cite{RizaldiKHFIAHN17,RizaldiA15} formalizes a set of traffic
rules for highway scenarios in Isabelle/HOL. It shows that traffic
rules can be used as requirements to be met by autonomous vehicles and
proposes a verification procedure.  A formalization of traffic rules
for uncontrolled intersections is provided in \cite{KarimiD20}, which
shows how the rules can be used by a simulator to safely control
traffic at intersections. The work in \cite{EsterleGK20} proposes a
methodology for formalizing traffic rules in linear temporal logic; it
shows how the evaluation of formalized rules on recorded human
behaviors provides insight into how well drivers follow the rules.

Many works deal with the formal verification of controllers that
perform specific maneuvers.  For example, in \cite{HilscherLOR11}, a
dedicated multi-way spatial logic inspired by interval temporal logic
is used to specify safety and provide proofs for lane change
controllers. The work in \cite{RizaldiISA18} presents a formally
verified motion planner in Isabelle/HOL. The planner uses maneuver
automata, a variant of hybrid automata, and linear temporal logic to
express properties. In \cite{EsterleGK20}, runtime verification is
applied to check that the maneuvers of a high-level planner conform to
traffic rules expressed in linear temporal logic.

Of particular interest for this work are correct by construction
techniques where system construction is guided by a set of properties
that the system is guaranteed to satisfy. They involve either the
application of monolithic synthesis techniques or compositional
reasoning throughout a component-based system design process. There is
considerable work on controller synthesis from a set of system
properties usually expressed in linear temporal logic, see for example
\cite{Kress-GazitP08,WongpiromsarnKF11,WongpiromsarnTM12,SchwartingAR18,arxiv.2203.14110}. These are algorithmic
techniques extensively studied in the field of control. They consist of
restricting the controllable behavior of a system interacting with its
environment so that a set of properties are satisfied. Nonetheless, their
application is limited due to their high computational cost, which depends
in particular on the type of properties and the complexity of the system
behavior.

An alternative to synthesis is to achieve correctness by design as a
result of composing the properties of the system components.
Component properties are usually "assume-guarantee" contracts
characterizing a causal relationship between a component and its
environment: if the environment satisfies the "assume" part of the
contract, the state of the component will satisfy the "guarantee"
part, e.g. \cite{BenvenisteCNPRR18,Meyer92,ChatterjeeH07}.  The use of
contracts in system design involves a decomposition of overall system
requirements into contracts that provide a basis for more efficient
analysis and validation. In addition, contract-based design is
advocated as a method for achieving correctness by design, provided
that satisfactory implementations of the system can be
found \cite{abs-1909-02070}. There are a number of theoretical
frameworks that apply mainly to continuous or synchronous systems,
especially for analysis and verification
purposes \cite{abs-2012-12657,MavridouKGKPW21,SaoudGF21}.  They suffer
computational limitations because, in the general case, they involve the
symbolic solution of fixed-point equations, which restricts the
expressiveness of the contracts \cite{MavridouKGKPW21}. Furthermore,
they are only applicable to systems with a static architecture, which
excludes dynamic reconfigurable systems, such as autonomous
systems.

The paper builds on previous results \cite{abs-2109-06478} on a
logical framework for parametric specification of \ADS\ combining
models of the system’s static environment in the form of maps, and the
dynamic properties of its vehicles.  Maps are metric graphs whose
vertices represent locations and edges are labeled with segments that
can represent roads at different levels of abstraction, with
characteristics such as length or geometric features characterizing
their shape and size.

An \ADS\ model is a dynamic system consisting of a map and a set of
vehicles moving along specific routes. Its state can be conceived as
the distribution of vehicles on a map with their positions, speeds and
other kinematic attributes. For its movement, each vehicle has a safe
estimate of the free space in its neighborhood, according to
predefined visibility rules.  We assume that vehicle coordination is
performed by a \Runtime\ that, for given vehicle positions and speeds
on the map, can compute the free spaces on each vehicle's itinerary in
which it can safely move.

We consider without loss of generality, \ADS\ with a discretized
execution time step $\Delta t$. Knowing its free space, each vehicle
can move by adapting its speed in order to stay in this space, braking
if necessary in case of emergency. At the end of each cycle, taking
into account the movements of the vehicles, the \Runtime\ updates
their positions on the map. The cycle iterates by calculating the free
spaces from the new state.

We study a safe control policy for \ADS, which is correct by
design. It results from the combination of two types of
assume-guarantee contracts: one contract for each vehicle and another
contract for the \Runtime\ taking into account the positions of the
vehicles on the map.

The contract for a vehicle states that, assuming that initially the
dynamics of the vehicle allow it to stay in the allocated free space,
it will stay in this free space. Note that the details of the contract
implementation are irrelevant; only the I/O relationship between free
space and vehicle speed matters.  The \Runtime\ assume-guarantee
contract asserts that if the free spaces allocated to vehicles at the
beginning of a cycle are disjoint, then they can be allocated new
disjoint free spaces provided they have fulfilled their contract.  The
combination of these two contracts leads to a control policy that
satisfies an inductive invariant, implying system safety.

We build on this general result by specializing its application in two
directions.  First, we show that it is possible to define speed
policies for vehicles that satisfy their assume-guarantee
contract. Second, we show that it is possible to define free space
policies for the \Runtime\ enforcing safety constraints of a given set
of traffic rules. We formalize traffic rules as a class of properties of
a linear temporal logic. Each rule applicable to a vehicle $c$, is an
implication whose conclusion involves two types of constraints on the
speed of $c$: speed regulation and speed limitation constraints. We
provide a method that derives from a given set of traffic rules,
constraints on the free spaces chosen by the \Runtime\ such that the
resulting system behavior is safe with respect to these rules.
This is the main result of the paper establishing correctness by
design of general \ADS, provided that their components comply with
their respective contracts.

The paper is structured as follows.  In Section \ref{sec:maps}, we
establish the general framework by introducing the basic models and
concepts for the representation of maps. In
Section \ref{sec:dynamic-model}, we introduce the dynamic model of \ADS\
involving a set of vehicles and a \Runtime\ for their coordination. We show
how a correct by design safe control policy is obtained by combining
assume-guarantee contracts for the vehicles and the \Runtime.  In
Section \ref{sec:speed-policies}, we study the principle of speed policies
respecting the vehicle contract and show its application through an
example.  In Section \ref{sec:freespace-policies}, we formalize traffic
rules as a class of formulas of a linear temporal logic and show how it is
possible to generate from a set of traffic rules free space policies such
that the system is safe by design.  In Section \ref{sec:experiments}, we
briefly describe the implementation of the approach and experiments
underway. Section \ref{sec:discussion} concludes with a discussion of the
significance of the results, future developments and applications.

\section{Map Representation}\label{sec:maps}

Following the idea presented in \cite{abs-2109-06478}, we build contiguous
road segments from a set $\Segments$ equipped with a partial concatenation
operator $\cdot : \Segments \times \Segments{}
\rightarrow \Segments \cup \{\bot\}$, a length norm $\segnorm{.} :
\Segments \rightarrow \PosReals$ and a partial subsegment extraction
operator $\subseg{.}{.}{.}:\Segments \times \PosReals \times \PosReals
\rightarrow \Segments \cup \{\bot\}$.
Thus, given a segment $s$, $\segnorm{s}$ represents its length and
$\subseg{s}{a}{b}$ for $0 \le a<b \le \segnorm{s}$, represents the
sub-segment starting at length $a$ from its origin and ending at length
$b$.  Segments can be used to represent roads at different levels of
abstraction, from intervals to regions. As an example, we consider
$\Segments$ as the set of curves obtained by concatenation of parametric
line segments and circle arcs. More precisely, for any $a,
r\in \PosReals^*$, $\varphi\in
\Reals$, $\theta \in \Reals^*$ the curves $line[a,\varphi]$,
$arc[r,\varphi,\theta]$ are defined as
$$\begin{array}{rcl}
    line[a,\varphi](t) & \eqdef & (at\cos\varphi,at\sin\varphi) ~~\forall t\in[0,1]\\
    arc[r,\varphi,\theta](t) & \eqdef & (r(\sin (\varphi + t\theta) - \sin\varphi),
    r (-\cos(\varphi + t \theta) + \cos\varphi)) ~~\forall t\in [0,1] 
  \end{array}$$
Note that $a$ and $r$ are respectively the length of the line and the
radius of the arc, $\varphi$ is the slope of the curve at the initial
endpoint and $\theta$ is the degree of the arc.
Fig.~\ref{fig:ex-curves} illustrates the composition of three curves
of this parametric form.

  \begin{figure}[htbp]
    \vspace{-3mm}
    \begin{center}
      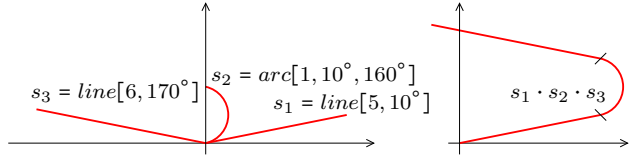
    \end{center}
    \vspace{-3mm}
    \caption{\label{fig:ex-curves}{Curve segments and their composition}}
    \vspace{-3mm}
  \end{figure}

We use metric graphs $G \eqdef (V,\Segments,E)$ to represent maps, where
$V$ is a finite set of \emph{vertices}, $\Segments$ is a set of segments
and $E \subseteq V \times \Segments^\star \times V$ is a finite set of
\emph{edges} labeled by \emph{non-zero length}
segments (denoted $\Segments^\star$).  For an edge $e=(v,s,v') \in E$ we
denote $\pre{e} \eqdef v$, $\post{e} \eqdef v'$, $\segment{e} \eqdef s$.
For a vertex $v$, we define $\pre{v} \eqdef
\{ e ~|~ \post{e} = v\}$ and $\post{v} \eqdef \{e ~|~ \pre{e} = v\}$.  We
call a metric graph \emph{connected} (resp. \emph{weakly connected}) if a
path (resp. an undirected path) exists between any pair of vertices.

We consider the set $\Pos{G} \eqdef V \cup \{(e,a) ~|~ e \in E,~ 0 \le a
\le \segnorm{\segment{e}}\}$ of \emph{positions} defined by a metric graph.
Note that positions $(e,0)$ and $(e, \segnorm{\segment{e}})$ are considered
equal respectively to positions $\pre{e}$ and $\post{e}$.  We denote by
$\edgeride{G}{p}{s}{p'}$ the existence of an $s$-labelled \emph{edge ride}
between succeeding positions $p=(e,a)$ and $p'=(e,a')$ in the same edge $e$
whenever $0 \le a < a' \le \segnorm{\segment{e}}$ and $s =
\subseg{\segment{e}}{a}{a'}$.  Moreover, we denote by $\ride{G}{p}{s}{p'}$
the existence of an $s$-labelled \emph{ride} between arbitrary positions
$p$, $p'$, that is, $\ride{G}{}{}{} \eqdef (\edgeride{G}{}{}{})^+$ the
transitive closure of edge rides.  Finally, we define the distance
$\distance{G}$ from position $p$ to position $p'$ as 0 whenever $p = p'$ or
the minimum length among all segments labeling rides from $p$ to $p'$ and
otherwise $+\infty$ if no such ride exists.  Whenever $G$ is fixed in the
context, we will omit the subscript $G$ for positions $\Pos{G}$, distance
$\distance{G}$, and rides $\edgeride{G}{}{}{}$ or $\ride{G}{}{}{}$.

A connected metric graph $G=(V,\Segments,E)$ can be interpreted as a map,
structured into roads and junctions, subject to additional assumptions:
\begin{itemize}
\item we restrict to metric graphs which are
  2D-consistent \cite{abs-2109-06478}, meaning intuitively they can be
  drawn in the 2D-plane such that the geometric properties of the segments
  are compatible with the topological properties of the graph. In
  particular, if two distinct paths starting from the same vertex $v$, meet
  at another vertex $v'$, the coordinates of $v'$ calculated from each path
  are identical.  For the sake of simplicity, we further restrict to graphs
  where distinct vertices are located at distinct points in the plane, and
  moreover, where no edge is self-crossing (meaning actually that distinct
  positions $(e,a)$ of the same edge $e$ are located at distinct points).
\item the map is equipped with a symmetric {\em junction}
  relationship $\crossing$ on edges $E$ which abstracts the geometric
  crossing (or the proximity) between edges at positions other than the
  edge end points.  This relationship is used to define
  the \emph{junctions} of the map, that is, as any non-trivial equivalence
  class in the transitive closure of $\crossing$.  Actually, junctions need
  additional signalisation to regulate the traffic on their edges (e.g.,
  traffic lights, stop signs, etc).  In addition, we assume a partial
  ordering $\prec_j$ on the set of vertices to reflect their static
  priorities as junction entries.
\item to resolve conflicts at merger vertices, i.e., vertices with two or
  more incident segments which do not belong to a junction, we assume that
  the map is equipped with a static priority relationship.  Specifically,
  for a vertex $v$, there is a total priority order $\priority{v}$ on the
  set of edges $\pre{v}$.  This order reflects an abstraction of the static
  priority rules associated with each of the merging edges (e.g.,
  right-of-way, yield-priority, etc).
\item every edge $e$ is associated with a maximal speed limit
  $\speed{e} \in \PosReals$.
\end{itemize}

  \begin{figure}[htbp]
    \vspace{-3mm}
    \begin{center}
      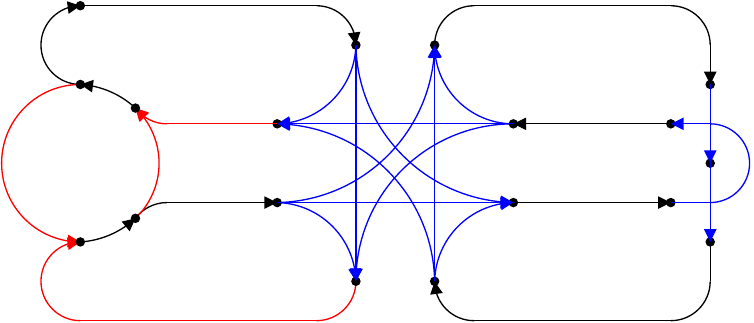
    \end{center}
    \vspace{-3mm}
    \caption{\label{fig:map}{A map with junctions (blue edges) and a merger vertices (red
    edges)}}
    \vspace{-3mm}
  \end{figure}

In the remainder of the paper, we consider a fixed metric graph $G=(V,
\Segments,E)$ altogether with the junction relationship $\crossing$, static priorities
$\priority{v}$ and edge speed limits as discussed above. Also, we extend
the junction and priority relationships from edges to their associated
positions, that is, consider $(e_1,a_1) \sim (e_2,a_2) \eqdef e_1 \sim e_2$
for any relation $\sim \in \{ \crossing, (\priority{v})_{v\in V} \}$.
Finally, we denote by $r_1 \uplus r_2$ the property that rides $r_1$, $r_2$
in $G$ are \emph{non-crossing}, that is, their sets of positions are
disjoint and moreover not belonging to the same junction(s), except for
endpoints.

\section{The \ADS\ Dynamic Model}\label{sec:dynamic-model}

\subsection{General \ADS\ Architecture}\label{subsec:architecture}

Given a metric graph $G$ representing a map, the state of an \ADS\ is a
tuple $\langle st_o \rangle_{o \in \Objects}$ representing the distribution
of a finite set of objects $\Objects$ with their relevant dynamic
attributes on the map $G$.  The set of objects $\Objects$ includes a set of
vehicles $\Vehicles$ and fixed equipment such as lights, road signs, gates,
etc.  For a vehicle $c$, its state $st_c \eqdef \langle \po{c}, \dt{c},
\speed{c}, \wt{c}, \itn{c} \dots \rangle$ includes respectively its
\emph{position} on the map (from $\Pos{}$), its \emph{displacement}
traveled since $\po{c}$ (from $\PosReals$), its \emph{speed} (from
$\PosReals$), the \emph{waiting time} (from $\PosReals$) which is the time
elapsed since the speed of $c$ became zero, its \emph{itinerary} (from the
set of segments $\Segments$) which labels a ride starting at $\po{c}$, etc.
For a traffic light $lt$, its state $st_{lt} \eqdef \langle
\mathit{\po{lt}, \att{lt}{cl},\dots} \rangle$ includes respectively its
$\emph{position}$ on the map (from $\Pos{}$), and its \emph{color} (with
values \emph{red} and \emph{green}), etc.

The general \ADS\ model is illustrated in Fig.~\ref{fig:archi} and
consists of a set of vehicle models $\Vehicles$ and a \Runtime\ that
interact cyclically with period $\Delta t$.  The \Runtime\ calculates
free space values for each vehicle $c$ which are lenghts $\fs{c}$ of
initial rides on their itineraries $\itn{c}$ whose positions are free
of obstacles.  In turn, the vehicles adapt their speed to stay within
the allocated free space. Specifically, the interaction proceeds as
follows:
\begin{itemize}
\item each vehicle $c$ applies a \emph{speed policy} for period $\Delta t$
  respecting its free space $\fs{c}$ received from the \Runtime.  During
  $\Delta t$, it travels a distance $\dt{c}'$ to some new position
  $\po{c}'$, and at the end of the period its speed is $\speed{c}'$, its
  itinerary $\itn{c}'$, etc. The new state is then communicated to the
  \Runtime.
\item the \Runtime\ updates the system state on the map taking into
  account the new vehicle states and time-dependent object attributes.
  Then it applies a \emph{free space policy} computing the tuple
  $\langle \fs{c}' \rangle_{c \in \Vehicles}$, the new free space for
  all vehicles based on the current system state.  The corresponding
  free spaces are then communicated to vehicles and the next cycle
  starts.  
\end{itemize}

  \begin{figure}[htbp]
    \vspace{-3mm}
    \begin{center}
      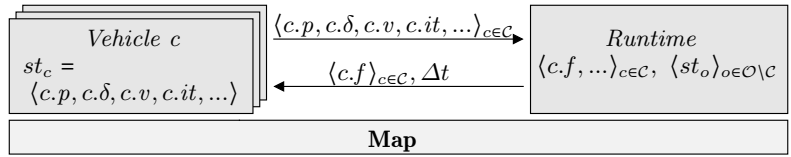
    \end{center}
    \vspace{-3mm}
    \caption{\label{fig:archi}{ General \ADS\ architecture}}
    \vspace{-3mm}
  \end{figure}

Note that the coordination principle described is independent of the
type of segments used in the map, e.g. intervals, curves or
regions. For simplicity, we take the free spaces to measure 
the length of an initial ride without obstacles on the vehicle
itinerary. This abstraction is sufficient to state the basic
results. We discuss later how they can be generalized for richer
interpretations of the map.

\subsection{Assume-Guarantee for Safe Control Policies}\label{subsec:ag-contracts}

We give below the principle of a safe control policy for vehicles, which
respects their allocated free space, applying assume-guarantee reasoning.

We consider the following two hypotheses. For a vehicle $c$, there exists a
function $\Brake_c : \PosReals \rightarrow \PosReals$ that gives the
minimum braking distance $c$ needs to stop from speed $v$, in case of
emergency. Furthermore, for a non-negative distance $f$, let $\Ahead{c}{f}$
denote the ride consisting of the positions reachable on the itinerary
$\itn{c}$ from the current vehicle position $\po{c}$ within distance $f$,
formally $\Ahead{c}{f} \eqdef \{ p' \in \Pos{} ~|~ \exists \delta \le
f. ~\ride{}{\po{c}}{\subseg{\itn{c}}{0}{\delta}}{p'} \}$.

The following definition specifies a safe control policy using
assume-guarantee reasoning on the components of the
\ADS\ architecture.  We consider assume-guarantee contracts on
components defined as pairs of properties $A/G$ specifying
respectively the input-output component behavior for a cycle, i.e.,
respectively, what the component guarantees ($G$) provided its
environment conforms to given assumption ($A$).

The policy is the joint enforcement of safe speed policies for vehicles and
safe free space policies for the \Runtime.  Vehicle safe speed policies
require that if a vehicle can brake safely by moving forward within its
allocated free space at the beginning of a cycle, then it can adapt its
speed moving forward within this space. \Runtime\ safe free space policies
require that if the free spaces of the vehicles are non-crossing at the
beginning of a cycle, then it is possible to find new non-crossing free spaces
for the vehicles provided they move forward in their allocated free space.

\begin{definition}[safe control policy]
  \label{def:safe-control-policy}
  A control policy is safe if
  \begin{itemize}
  \item each vehicle $c \in \Vehicles$ respects the assume-guarantee contract:
    \begin{eqnarray*}
      0 \le \speed{c},~B_c(\speed{c}) \le \fs{c} & \big/ &
      0 \le \speed{c}',~ 0 \le \dt{c}',~ \dt{c}' + B_c(\speed{c}') \le \fs{c},~ \\
      & & \ride{}{\po{c}}{\subseg{\itn{c}}{0}{\dt{c}'}}{\po{c}'},~ \itn{c}' = \subseg{\itn{c}}{\dt{c}'}{-}
    \end{eqnarray*}
  \item the \Runtime\ respects the assume-guarantee contract:
    \begin{eqnarray*}
      \wedge_c 0 \le \dt{c} \le \fs{c},~ \uplus_c \Ahead{c}{\fs{c} - \dt{c}} & \big/ &
      \wedge_{c \in \Vehicles} \fs{c}' \ge \fs{c} - \dt{c},~ \uplus_{c\in \Vehicles} \Ahead{c}{\fs{c}'}
      \end{eqnarray*}
  \end{itemize}
\end{definition}

\begin{theorem}\label{thm:safe-control-policy}
  Safe control policies preserve the following invariants:
  \begin{itemize}
  \item the speed is positive and compliant to the free space, for all vehicles, that is, \\
    $\bigwedge_{c \in \Vehicles} 0 \le \speed{c} \wedge B(\speed{c}) \le \fs{c}$,
  \item the free spaces are non-crossing, that is, 
    $\biguplus_{c \in \Vehicles} \Ahead{c}{\fs{c}}$.
  \end{itemize}
\end{theorem}
\begin{proof}
  Consider the usual notation for Hoare triples $\{ \phi \} P \{ \psi
  \}$ denoting that whenever the precondition $\phi$ is met, executing
  the program $P$ establishes the postcondition $\psi$.  Let
  $A_P(X)/G_P(X,X')$ be a contract for a program $P$ whose initial
  state $X$ satisfies the assumption $A_P(X)$ and which when it
  terminates guarantees the relation $G_P(X,X')$ between $X$ and the
  final state $X'$.  Then, Hoare triples are established
  by the proof rule:
  \[\begin{array}{c}
  \phi(X) \implies A_P(X) ~~~~ \exists X.~\phi(X) \wedge G_P(X,X') \implies \psi(X') \\ \hline
  \{ \phi \} P \{ \psi \}
  \end{array} \]
  We now prove that the conjunction of two assertions in the theorem is an
  inductive invariant, holding at the beginning of every cycle.  First,
  using the rule above for the assume-guarantee contract on \emph{all}
  vehicles we establish the following triple, where $||_{c \in \Vehicles} c$
  represents the program executed by the vehicle controllers in one cycle:
  \[ \begin{array}{l}
    \{ \bigwedge_{c \in \Vehicles} 0 \le \speed{c} \wedge B(\speed{c}) \le \fs{c} \wedge
    \biguplus_{c \in \Vehicles} \Ahead{c}{\fs{c}} \} \\
    \hspace{1cm} ||_{c \in \Vehicles} ~c  \\
    \{ \bigwedge_{c \in \Vehicles} 0 \le \speed{c} \wedge 0 \le \dt{c} \wedge
    \dt{c} + B(\speed{c}) \le \fs{c} \wedge
     \biguplus_{c \in \Vehicles} \Ahead{c}{\fs{c} - \dt{c}} \} \\
  \end{array} \]
  The arithmetic constraints on the speed, distance traveled and free
  space are implied from the guarantee.  The constraint on the free space
  takes into account the update of the vehicle positions, that is, moving ahead
  into their free space by the distance traveled.
  Second, using the assume-guarantee contract on the \Runtime\ we establish the triple
  \[ \begin{array}{l}
    \{ \bigwedge_{c \in \Vehicles} 0 \le \speed{c} \wedge 0 \le \dt{c} \wedge
    \dt{c} + B(\speed{c}) \le \fs{c} \wedge
    \biguplus_{c \in \Vehicles} \Ahead{c}{\fs{c} - \dt{c}} \} \\
    \hspace{1cm} \mbox{\Runtime} \\
    \{ \bigwedge_{c \in \Vehicles} 0 \le \speed{c} \wedge B(\speed{c}) \le \fs{c} \wedge
    \biguplus_{c \in \Vehicles} \Ahead{c}{\fs{c}} \} \\
  \end{array} \]
  That is, the \Runtime\ re-establishes the invariant essentially by
  providing at least the same free space as in the previous cycle.\qed
\end{proof}

Note that this theorem guarantees the safety of the coordination
insofar as the vehicles respecting their contracts remain in their
allocated free spaces which are non-crossing by
construction. Nevertheless, the result leaves a lot of freedom to
vehicles and the \Runtime\ to choose speeds and non-crossing free spaces. In
particular, two questions arise concerning these choices. The first
question is wether the system can reach states where no
progress is possible.  One can imagine traffic jam situations, for
example when vehicles do not have enough space to move.  The second
question is whether free space choices can be determined by traffic
rules that actually enforce fairness in resolving conflicts between
vehicles. This question is discussed in detail in Section
\ref{sec:freespace-policies}.

We show below that it is possible to compute non-blocking control
policies by strengthening the contracts satisfied by the vehicles and
the \Runtime\ with additional conditions. For vehicles, we require that
they move in a cycle if their free space is greater than a minimum
free space $f_{min}$. This constant should take into account the
dimensions of the vehicles and their dynamic characteristics, e.g.,
the minimum space needed to safely reach a non-negative speed from a
stop state.  Additional conditions for the contract of the \Runtime\ are
that if all vehicles are stopped, then it can find at least one free
space greater than $f_{min}$.

\begin{definition}[non-blocking control policy]
  \label{def:non-blocking-control-policy}
  A control policy is non-blocking if there exists non-negative
  $f_{min}$ such that:
  \begin{itemize}
  \item each vehicle $c\in C$ respects the $A/G$ contract:
    $ \fs{c} \ge f_{min} ~\big/~ \speed{c}' > 0$,
  \item the \Runtime\ respects the $A/G$ contract:
    $ \bigwedge_{c} \speed{c} = 0 ~\big/~ \max_c \fs{c}' \ge f_{min}$.
  \end{itemize}
\end{definition}
\begin{theorem}
  Non-blocking control policies ensure progress i.e., there is always
  a vehicle whose speed is positive.
\end{theorem}
\begin{proof}
  The proof is an immediate consequence of the two contracts of
  Def.~\ref{def:non-blocking-control-policy}.  If all vehicles stop moving
  during a cycle, the \Runtime\ will necessarily find at least $f_{\min}$
  free space for at least one of them.  Then, at the next cycle, at least
  one vehicle will move again with positive speed, which concludes the
  proof.  \qed
\end{proof}

\section{Speed Policies abiding by the Vehicle Contract}\label{sec:speed-policies}
In this section, we show that it is possible for vehicles to compute
speed policies in accordance with their contract.

The behavior of each vehicle is defined by a controller, which given its
current speed and its free space, computes the displacement for
$\Delta t$ so that it can safely move in the free space.  Such safe speed policies
have been studied in \cite{WangLS20,abs-2103-15484}.

We illustrate the principle of safe speed policy with respect to $f$
considering that each vehicle is equipped with a controller that
receives a free space value and adjusts its speed adequately.  For the
sake of simplicity, assume the controller can select among three
different constant acceleration values $\{-b_{max}, 0,
a_{max}\} \in \Reals$ respectively, the negative value $-b_{max}$ for
decreasing, the zero value for maintaining and the positive value
$a_{max}$ for increasing the speed.  At every cycle, the controller
will select the highest acceleration value for which the vehicle
guarantee holds as defined by its contract in
Def.~\ref{def:safe-control-policy}.  Nonetheless, an exception applies
for the very particular case where the vehicle stops within the cycle,
which cannot be actually handled with constant acceleration.

The proposed speed policy defines the new speed $v'$ and
displacement $\delta'$ using a region decomposition of the safe $v \times
f$ space (that is, where $v \ge 0$ and $f \ge \Brake(v)$) as follows:
\begin{equation}
v',~ \delta' \eqdef \left\{ \begin{array}{l l}
  0,~ f & \mbox{if } f \ge \Brake(v), f - v \Delta t < \Brake(v), v - b_{max} \Delta t < 0 \\[5pt]
  \multicolumn{2}{l}{v -b_{max}\Delta t,~ v \Delta t - b_{max} \Delta t^2 / 2} \\
  &  \mbox{if } f \ge \Brake(v), f - v \Delta t < \Brake(v), v - b_{max} \Delta t \ge 0 \\[5pt]
  v,~ v \Delta t & \mbox{if } f - v \Delta t \ge \Brake(v), f - v\Delta t - a_{max} \Delta t^2 / 2 < \Brake(v + a_{max}\Delta t) \\[5pt]
  \multicolumn{2}{l}{v + a_{max}\Delta t,~ v \Delta t + a_{max} \Delta t^2 / 2} \\
  & \mbox{if } f - v\Delta t - a_{max} \Delta t^2 / 2 \ge \Brake(v + a_{max}\Delta t)
\end{array} \right.
\end{equation}
Intuitively, the regions are defined such that, when the corresponding
acceleration is constantly applied for $\Delta t$ time units, the guarantee
on the vehicle is provable given the assumptions and the region boundary
conditions.  For illustration, the regions are depicted in
Fig.~\ref{fig:region-policy} for some concrete values of $\Delta t = 1~sec$,
$a_{max}=2.5 m/s^2$ and $b_{max} = -3.4 m/s^2$.

\begin{figure}[htbp]
  \vspace{-3mm}
  \begin{center}
    \includegraphics[width=.75\textwidth]{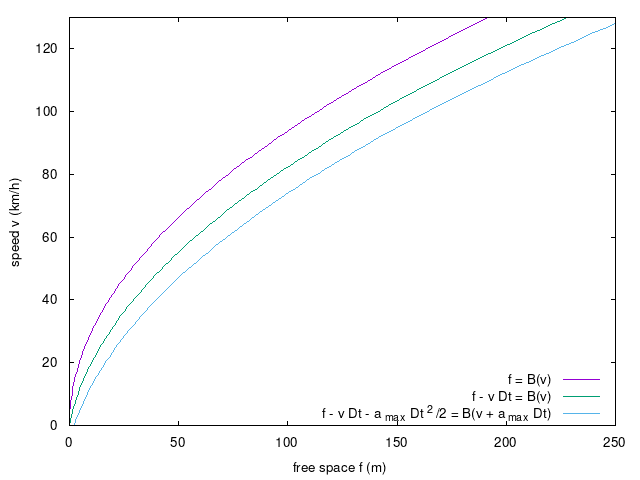}
  \end{center}
  \vspace{-3mm}
  \caption{\label{fig:region-policy} Region decomposition for safe speed policy}
  \vspace{-3mm}
\end{figure}

Moreover, the vehicle position and the itinerary are updated according to
the travelled distance by taking $\po{c}'$ such that
$\ride{}{\po{c}}{\subseg{\itn{c}}{0}{\dt{c}'}}{\po{c}'}$ and $\itn{c}'
= \subseg{\itn{c}}{\dt{c}'}{-}$.  Furthermore, the waiting time $\wt{c}$ is
updated bt taking $\wt{c}' \eqdef \wt{c} + \Delta t$ if $\speed{c}
= \speed{c}' = 0$ and $\wt{c}' \eqdef 0$ otherwise.

\begin{proposition}
  The region-based speed policy respects the safety  contract for vehicles if the
  braking function is $\Brake(v) = v^2 / 2 b_{max}$.
\end{proposition}

\begin{proof}
According to Def.~\ref{def:safe-control-policy} let assume that $0 \le v$
and $\Brake(v) \le f$.  The policy must guarantee that $0 \le v'$,
$0 \le \delta'$, $\delta' + \Brake(v') \le f$,
$\ride{}{\po{c}}{\subseg{\itn{c}}{0}{\dt{c}'}}{\po{c}'}$, $\itn{c}'
= \subseg{\itn{c}}{\dt{c}'}{-}$.  Obviously, the last two constraints are
explicitely enforced.  For the remaining constraints, the proof is made on
a case by case basis for the four regions:

\begin{enumerate}[label=(\roman*)]
\item Immediate as $v' = 0$, $\delta' = f$ in this case.
  
\item The constraint $v' \ge 0$ is equivalent to $v - b_{max} \Delta
    t \ge 0$ which is one of the region boundaries.  The constraint
    $\delta' \ge 0$ holds as $ \delta' = v\Delta t - b_{max} \Delta t^2
    = \Delta t (v - b_{max} \Delta t/2)$.  Consider the constraint $\delta'
    + \Brake(v') \le f$.  The term $\delta' + \Brake(v')$ can be
    successively rewritten as follows:
    \[\begin{array}{rcl}
      \delta' + \Brake(v') & = & v \Delta t - b_{max} \Delta t^2 / 2 + \Brake(v - b_{max} \Delta t) \\
      & = &  v \Delta t - b_{max} \Delta t^2 / 2 + (v - b_{max} \Delta)^2 / 2 b_{max} \\
      & = & v \Delta t - b_{max} \Delta t^2 / 2 + v^2 / 2 b_{max} - v \Delta t + b_{max} \Delta t^2 / 2 \\
      & = & v^2 / 2 b_{max} = \Brake(v)
    \end{array}\]
    Henceforth, $\delta' + \Brake(v') \le f$ is equivalent to $\Brake(v)
    \le f$ and holds from the assumption.
    
  \item The constraint $v' \ge 0$ holds because $v'=v$.  The constraint
    $\delta' \ge 0$ holds because $\delta' = v \cdot \Delta t$. The
    constraint $\delta' + \Brake(v') \le f$ is equivalent to $v \Delta t
    + \Brake(v) \le f$ which is one of the region boundaries.
    
  \item The constraint $v' \ge 0$ holds because $v' = v + a_{max} \Delta t$
    and $v$, $a_{max}$, $\Delta t$ are all positive.  Similarly, the
    constraint $\delta' \ge 0$ holds as $\delta' = v \Delta t +
    a_{max} \Delta t^2/2$.  The constraint $\delta' + \Brake(v') \le f$ is
    equivalent to $v \Delta t + a_{max} \Delta t^2 / 2 + \Brake(v +
    a_{max}\Delta t) \le f$ and is the region boundary.\qed
\end{enumerate}
\end{proof}

\begin{proposition}
The region-based policy respects the non-blocking contract for vehicles for any
$f_{min} \ge \Brake(a_{max}\Delta t) + a_{max}\Delta t^2/2$.
\end{proposition}
\begin{proof}
Actually, the value $\Brake(a_{max}\Delta t) + a_{max}\Delta t^2/2$
represents the minimal amount of space for which the policy will select
acceleration $a_{max}$ when the speed is zero, as defined by the condition
of the 4th region. \qed
\end{proof}

Note that the speed policy works independently of the value of the parameter
$\Delta t$, which is subject only to implementation constraints, e.g.,
it must be large enough to allow the controlled electromechanical
system to realize the desired effect. A large $\Delta t$ may imply low
responsiveness to changes and jerky motion, but will never compromise
the safety of the system.

The proposed implementation of the speed policy is "greedy" in the
sense that it applies maximum acceleration to move as fast as possible
in the available space. We could have "lazy" policies that do not move
as fast as possible, and simply extend the travel time.  We have shown
in \cite{WangLS20} that the region-based speed policy approaches the
optimal safety policy, i.e., the one that gives the shortest travel
time, when we refine the choice of acceleration and deceleration rates
in the interval $[-b_{max}, a_{max}]$.

\section{Free Space Policies implied by Traffic Rules}\label{sec:freespace-policies}
In this section, we study free space safety policies for a given set of
global system properties describing traffic rules. We formalize traffic
rules as a class of linear temporal logic formulas and provide a method for
computing free space values for vehicles that allow them to meet a given
set of traffic rules.

\subsection{Writing Specifications of Traffic Rules}

Traffic rules are a special class of properties that can be reliably
applied by humans.  They involve the responsibility of a driver who can
control the speed and direction of a vehicle on the basis of an
approximate knowledge of its kinetic state.  They do not include conditions
that are difficult to assess by the subjective judgment of human drivers,
whereas they could be verified by properly instrumented computers.  Thus,
traffic rules are based on topological considerations rather than on
quantitative information, such as an accurate comparison of vehicle
speeds. Furthermore, they can be formulated as implications where the
implicant is a condition that is easy to check by a driver and the
conclusion is a constraint on controllable variables that call for a
possible corrective action by the driver.

Given a map $G$ and a set of objects $\Objects$, we specify traffic
rules as formulas of a linear time logic of the following form, where
$\always$ is the \emph{always} time modality and $\next$ is the \emph{next} time
modality:
\begin{equation}
  \label{eq:traffic-rule}
  \always~ \forall c_1.~\forall o_2...\forall o_k.~ \phi(c_1,o_2,\dots,o_k) \implies
  \next \psi(c_1,o_2,\dots,o_k)
\end{equation}
A rule says that for any run of the system, the satisfaction of the
precondition $\phi$ implies that the postcondition $\psi$ holds at the
next state. Both $\phi$ and $\psi$ are boolean combinations of state
predicates as defined below. Furthermore, we assume that $\psi$
constrains the speed of a single vehicle $c_1$ for which the property
is applicable, and which we call for convenience the \emph{ego}
vehicle.

The rules involve state predicates $\phi$ in the form of 
first-order assertions built
from variables and object attributes (denoting map positions, segments,
reals, etc) using available primitives on map positions (e.g., rides
$\ride{}{}{}{}$, edge rides $\edgeride{}{}{}{}$, distance $\distance{}$,
equality $=$), on segments (e.g., concatenation and subsegment
extraction), in addition to real arithmetic and boolean operators.

Moreover, we define auxiliary non-primitive \emph{location} and
\emph{itinerary} predicates proven useful for the expression of traffic
rules. For a vehicle $c \in \Vehicles$ and $x$ either an object $o\in
\Objects$, a vertex $u$ or an edge $e$ of the map, we define the predicates
$\at{c}{x}$ ($c$ \emph{is at} $x$), $\edgemeets{c}{x}$ ($c$ \emph{meets}
$x$ \emph{along the same edge}), $\meets{c}{x}$ ($c$ \emph{meets} $x$) as
in Table~\ref{tab:aux-predicates}.  Furthermore, for a vehicle $c\in
\Vehicles$ and non-negative $\delta$ let $\po{c} \oplus_c \delta$ denote
the future position of $c$ after traveling distance $\delta$, that is,
either $\po{c}$ if $\delta=0$ or the position $p'$ such that
$\ride{}{\po{c}}{\subseg{\itn{c}}{0}{\delta}}{p'}$. We extend $\oplus_c$ to
arbitrary future positions of $c$ by taking $(\po{c} \oplus_c \delta)
\oplus_c \delta' \eqdef \po{c} \oplus_c (\delta + \delta')$ and we
consider the total ordering $\le_c$ defined
as $\po{c} \oplus_c \delta \le_c \po{c} \oplus_c \delta'$ if and only if $\delta \le
\delta'$.
 
\begin{table}[htbp]
  \[ \begin{array}{c|c|c|c} \hline
    &
    \at{c}{x} &
    \edgemeets{c}{x} &
    \meets{c}{x}
    \\ \hline
    x = o &
    \po{c} = \po{o} & 
    \exists \delta.~ \edgeride{}{\po{c}}{\subseg{\itn{c}}{0}{\delta}}{\po{o}} &
    \exists \delta.~ \ride{}{\po{c}}{\subseg{\itn{c}}{0}{\delta}}{\po{o}} 
    \\ 
    x = u &
    \po{c} = u &
    \exists \delta.~ \edgeride{}{\po{c}}{\subseg{\itn{c}}{0}{\delta}}{u} &
    \exists \delta.~ \ride{}{\po{c}}{\subseg{\itn{c}}{0}{\delta}}{u} 
    \\
    x = e &
    \exists a.~\po{c} = (e,a) &
    \exists \delta.~\exists a>0.~ \edgeride{}{\po{c}}{\subseg{\itn{c}}{0}{\delta}}{\pre{e}} \wedge &
    \exists \delta.~\exists a>0.~ \ride{}{\po{c}}{\subseg{\itn{c}}{0}{\delta}}{\pre{e}} \wedge
    \\
    &
    &
    \ride{}{\po{c}}{\subseg{\itn{c}}{0}{\delta+a}}{(e,a)} &
    \ride{}{\po{c}}{\subseg{\itn{c}}{0}{\delta+a}}{(e,a)}
    \\ \hline
  \end{array} \]
  \caption{\label{tab:aux-predicates} Location and itinerary predicates.}
  \vspace{-5mm}
\end{table}

We define the semantics of state predicates $\phi$ in the usual way, by
providing a satisfaction relation $\sigma, st \vdash \phi$, where $\sigma$
is an assignment of free variables of $\phi$ and $st$ is a system state.  A
complete formal definition can be found in \cite{abs-2109-06478}.
The semantics of rules is defined on pairs $\sigma,[st^{(t_i)}]_{i \ge
  0}$ consisting of a function $\sigma$ assigning objects instances to
object variables of the formulas and a run $[st^{(t_i)}]_{i \ge 0}$ for a
finite set of objects $\Objects$.  For initial state $st^{(t_0)}$ we define
\emph{runs} as sequences of consecutive states $[st^{(t_i)}]_{i \ge 0}$
obtained along the cyclic \ADS\ execution as described in section
\ref{subsec:architecture} and parameterized by the increasing sequence of
time points $t_i \eqdef t_0 + i \cdot \Delta t \in \PosReals$, that is,
equal to the time for reaching the $i^{th}$ system state.

\begin{table}[t!]
  \[ \begin{array}{cl} \\ \hline
    1 & \mbox{enforcing safety distance between following vehicles $c_1$ and $c_2$:} \\
    & \always~ \forall c_1.~\forall c_2.~\meets{c_1}{c_2} \implies
    \next \Brake_{c_1}(\speed{c_1}) \le \distance{}(\po{c_1},\po{c_2}) \\[5pt]
    
    2 & \mbox{coordination within all-way-stop junctions:} \\
    (i) & \mbox{safe braking of vehicle $c_1$ approaching a stop $so_1$} \\
    & \always~ \forall c_1.~\forall so_1.~ \edgemeets{c_1}{so_1}  \implies
    \next \Brake_{c_1}(\speed{c_1}) \le \distance{}(\po{c_1},\po{so_1}) \\
    (ii) & \mbox{vehicle $c_1$ obeys a stop sign when another vehicle $c_2$ crosses the junction} \\
    & \always~ \forall c_1.~\forall so_1.~ \forall c_2.~  \at{c_1}{so_1} \wedge \speed{c_1}=0
    \wedge \speed{c_2} > 0  \wedge  \po{c_1} \crossing \po{c_2} \implies \next \speed{c_1} = 0 \\
    (iii) & \mbox{if two vehicles $c_1$, $c_2$ are waiting before the respective stops $so_1$, $so_2$} \\
    & \mbox{and $c_2$ waited longer than $c_1$ then $c_1$ has to stay stopped} \\
    & \always~ \forall c_1.~\forall so_1.~ \forall c_2.~ \forall so_2.~ \at{c_1}{so_1} \wedge \speed{c_1} = 0
    \wedge \at{c_2}{so_2}  \wedge \speed{c_2} = 0 ~\wedge \\
    & \hspace{1cm} \po{c_1} \crossing \po{c_2}  \wedge \wt{c_1} < \wt{c_2} \implies \next \speed{c_1} = 0 \\
    (iv) & \mbox{if two vehicles $c_1$, $c_2$ are waiting before the respective stops $so_1$, $so_2$ the same} \\
    & \mbox{amount of time and $c_2$ is at an entry with higher priority then $c_1$ has to stay stopped} \\
    & \always~ \forall c_1.~\forall so_1.~ \forall c_2.~ \forall so_2.~ \at{c_1}{so_1} 
    \wedge \speed{c_1} = 0 \wedge \at{c_2}{so_2}  \wedge \speed{c_2} = 0 ~\wedge \\
    & \hspace{1cm} \po{c_1} \crossing \po{c_2} \wedge \wt{c_1} = \wt{c_2} \wedge \po{so_1} \prec_j \po{so_2}
    \implies \next \speed{c_1} = 0 \\[5pt]

    3 & \mbox{coordination using traffic-lights:} \\
    & \mbox{if vehicle $c_1$ meets a red traffic light $lt_1$, it will remain in safe distance} \\
      & \always~ \forall c_1.~\forall lt_1.~ \edgemeets{c_1}{lt_1} \wedge \att{lt_1}{color}=red \wedge
    \Brake_{c_1}(\speed{c_1}) \le \distance{}(\po{c_1},\po{lt_1}) \\ 
    & \hspace{1cm} \implies \next \Brake_{c_1}(\speed{c_1}) \le \distance{}(\po{c_1},\po{lt_1}) \\[5pt]

    4 & \mbox{priority-based coordination of two vehicles $c_1$ and $c_2$ whose itineraries } \\
    & \mbox{meet at merger vertex $u$:} \\ 
    (i) & \mbox{if $c_2$ cannot stop at $u$ then $c_1$ must give way} \\
    & \always~ \forall c_1.~\forall c_2.~\forall u.~ \edgemeets{c_1}{u} \wedge
    \Brake_{c_1}(\speed{c_1}) \le \distance{}(\po{c_1},u) ~\wedge \\
    & \hspace{1cm} \edgemeets{c_2}{u} \wedge \Brake_{c_2}(\speed{c_2}) > \distance{}(\po{c_2},u) 
     \implies \next \Brake_{c_1}(\speed{c_1}) \le \distance{}(\po{c_1},u)\\
    (ii) & \mbox{if $c_1$, $c_2$ are reaching $u$ and $c_1$ has less priority than $c_2$ then $c_1$ must give way} \\
     & \always~ \forall c_1.~\forall c_2.~\forall u.~ \edgemeets{c_1}{u} \wedge
    \Brake_{c_1}(\speed{c_1}) = \distance{}(\po{c_1},u) \wedge 
     \po{c_1} \priority{u} \po{c_2} ~\wedge \\
    & \hspace{1cm} \edgemeets{c_2}{u} \wedge \Brake_{c_2}(\speed{c_2}) = \distance{}(\po{c_2},u)  \implies
    \next \Brake_{c_1}(\speed{c_1}) \le \distance{}(\po{c_1},u) \\[5pt]

    5 & \mbox{enforcing speed limits for vehicle $c_1$:} \\
    (i) & \mbox{if $c_1$ is traveling in an edge $e$ then its speed should be lower than the speed limit} \\
    & \always~ \forall c_1.~ \forall e.~ \at{c_1}{e} \implies \next \speed{c_1} \le \speed{e} \\
    (ii) & \mbox{if $c_1$ is approaching an edge $e$ then it controls its speed so that it complies} \\
    & \mbox{ with the speed limit at the entrance of $e$} \\
    & \always~ \forall c_1.~ \forall e.~ \edgemeets{c_1}{e} \implies
    \next \Brake_{c_1}(\speed{c_1}) \le \distance{}(\po{c_1},\pre{e}) + \Brake_{c_1}(\speed{e}) \\ \hline
  \end{array} \]
  \caption{\label{tab:traffic-rules} Traffic rules}
  \vspace{-5mm}
\end{table}  

We provide examples of traffic rules in Table~\ref{tab:traffic-rules}.
We restrict ourselves to safety rules that characterize boundary
conditions that should not be violated by the driver controlling the
vehicle speed. Therefore, the preconditions characterize potential
conflict situations occurring at intersections as well as other
constraints implied by the presence of obstacles or speed rules, e.g.,
traffic lights or speed limit signals. The preconditions may involve
various itinerary and location predicates and constraints on the speed
of the ego vehicle. Moreover, the latter are limited to constraints
maintained by the vehicle and involving braking functions in the form
$\Brake_c (\speed{c}) ~\#~ k$ where $k$ is a distance with respect to
a reference position on the map and $\#$ is a relational symbol $\#
\in \{<,\le,=,\ge,>\}$. Furthermore, the postconditions 
involve two types of constraints on the speed of the ego vehicle:
either speed regulation constraints that limit the distance to full
stop, that is $\Brake_{c_1}(\speed{c_1})$, or speed limitation
constraints requiring that the speed $\speed{c_1}$ does not exceed a given limit
value.

Note the difference with other approaches using unrestricted
linear temporal logic, with "eventually" and "until" operators, to
express traffic rules, e.g. \cite{abs-2109-06478}. We have adopted the above
restrictions because they closely characterize the vehicle safety
obligations in the proposed model. Furthermore, as we show below,
traffic rules of this form can be translated into free space
rules that can reinforce the policy managed by the \Runtime.

\subsection{Deriving Free Space Rules from Traffic Rules}

We show that we can derive from traffic rules limiting the speed of
 vehicles, rules on free space variables controlled by the \Runtime\ as illustrated in
Fig.~\ref{fig:refined-archi}) such that both the traffic rules and 
the free space contract hold.

  \begin{figure}[htbp]
    \vspace{-3mm}
    \begin{center}
      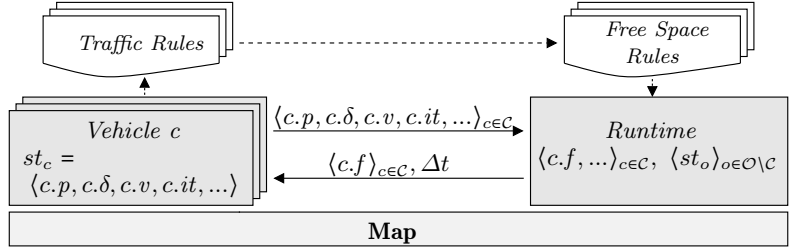
    \end{center}
    \vspace{-3mm}
    \caption{\label{fig:refined-archi}{Refined \ADS\ architecture}}
    \vspace{-3mm}
  \end{figure}

To express constraints on the free space variables $\fs{c}$, we use,
for vehicles $c$, auxiliary \emph{limit position} variables $\langle
\limit{c} \rangle_{c \in \Vehicles}$ such that $\limit{c} = \po{c}
\oplus_c \fs{c}$.  In other words, the limit position $\limit{c}$
defines the position beyond which a vehicle should not be according to its
contract. It is clear that for given $\limit{c}$ and $\po{c}$,
$\fs{c}$ is defined as the distance from $\po{c}$ to $\limit{c}$.

Using the limit position variables $\langle \limit{c} \rangle_{c \in
  \Vehicles}$ we can transform structurally any state formula $\phi$ into a
free space formula $\phi_\pi$ by replacing constraints on speeds by
induced constraints on limit positions as follows, for $\# \in \{<, \le, =,
\ge, >\}$ and $t$ a non-negative real constant:
\[ \begin{array}{rcl c rcrcl}
  \Brake_{c}(\speed{c}) & \# & \distance{}(\po{c},x) + t &
  \hspace{1cm} \mapsto \hspace{1cm} &
  \limit{c} & \#_c & x & \oplus_c & t \\
  \speed{c} & \# & t &
  \mapsto &
  \limit{c} & \#_c & \po{c} & \oplus_{c} & \Brake_c(t) 
\end{array} \]

The first case concerns speed regulation constraints 
bounding the limit position $\limit{c}$ relatively to the position $x$ of a
fixed or moving obstacle ahead of $c$, that is, a stop or traffic light
sign, a vehicle, etc.  The second case concerns speed limitation
constraints bounding $\limit{c}$ relatively to the current vehicle
position $\po{c}$ and the allowed speed.

Given a state formula $\phi$, let $\phi_{\pi}$ be the derived formula obtained by
replacing constraints on speeds by constraints on limit
positions. The following theorem guarantees preservation between properties
involving speed constraints and properties involving limit positions, in
relation to the vehicle speed contracts.

\begin{theorem}
  \label{thm:preservation}
  The following equivalences hold:
  \[ \begin{array}{rrcl}
    (i) & \phi & \Longleftrightarrow & (\exists~ \limit{c})_c~ \phi_\pi \wedge
    \bigwedge\nolimits_{c} \Brake_c(\speed{c}) = \distance{}(\po{c},\limit{c}) \\
    (ii) & \cl{\phi} & \Longleftrightarrow & (\exists~ \limit{c})_c~ \phi_\pi \wedge
    \bigwedge\nolimits_{c} \Brake_c(\speed{c}) \le \distance{}(\po{c},\limit{c})
  \end{array}\]
  where $\cl{\phi}$ is the speed-lower closure of $\phi$, that is,
  $\phi$ where speed constraints
  of the form $\speed{c} ~\#~ t$, $\Brake_c(\speed{c}) ~\#~
  \distance{}(\po{c},x)) + t$ for $\# \in \{ \ge, > \}$ are removed.
\end{theorem}
\begin{proof}
  (i) Assuming $\Brake_{c}(\speed{c}) = \distance{}(\po{c},\limit{c})$ the
  following equivalences between constraints on speed $\speed{c}$ and
  derived constraints on limit position $\limit{c}$ hold trivially, for any
  $\# \in \{<,\le,=,\ge,>\}$:
  \[ \begin{array}{rrcl}
    &
    \Brake_{c}(\speed{c}) ~\#~ \distance{}(\po{c},x) + t & \wedge &
    \Brake_{c}(\speed{c}) = \distance{}(\po{c},\limit{c}) \\
    \Longleftrightarrow &
    \distance{}(\po{c},\limit{c}) ~\#~ \distance{}(\po{c},x) + t & \wedge &
    \Brake_{c}(\speed{c}) = \distance{}(\po{c},\limit{c}) \\
    \Longleftrightarrow &
    \limit{c} ~\#~ x \oplus_c t & \wedge &
    \Brake_{c}(\speed{c}) = \distance{}(\po{c},\limit{c}) \\[5pt]
    
    & \speed{c} ~\#~ t & \wedge & 
    \Brake_{c}(\speed{c}) = \distance{}(\po{c},\limit{c}) \\
    \Longleftrightarrow &
    \Brake_{c}(\speed{c}) ~\#~ \Brake_c(t) & \wedge &
    \Brake_{c}(\speed{c}) = \distance{}(\po{c},\limit{c}) \\    
    \Longleftrightarrow &
    \distance{}(\po{c},\limit{c}) ~\#~ \Brake_c(t) & \wedge &
    \Brake_{c}(\speed{c}) = \distance{}(\po{c},\limit{c}) \\    
    \Longleftrightarrow &
    \limit{c} ~\#~ \po{c} \oplus_c \Brake_c(t) & \wedge &
    \Brake_{c}(\speed{c}) = \distance{}(\po{c},\limit{c}) \\    
  \end{array} \]
  This implies, assuming $\bigwedge_c \Brake_{c}(\speed{c}) =
  \distance{}(\po{c},\limit{c})$ that any state formula $\phi$ is equivalent to
  the derived constraint $\phi_\pi$ on limit positions, that is:
  \[ \phi \wedge \bigwedge\nolimits_c \Brake_{c}(\speed{c}) = \distance{}(\po{c},\limit{c})
  \Longleftrightarrow
  \phi_\pi \wedge \bigwedge\nolimits_c \Brake_{c}(\speed{c}) = \distance{}(\po{c},\limit{c}) \]
  Then, as $\phi$ does not involve limit positions we have:
  \[ \begin{array}{rcr}
    \phi & \Longleftrightarrow &
    \phi \wedge (\exists \limit{c})_c~ \bigwedge\nolimits_c
    \Brake_{c}(\speed{c})  = \distance{}(\po{c},\limit{c}) \\
    & \Longleftrightarrow &
    (\exists \limit{c})_c~ \phi \wedge \bigwedge\nolimits_c
    \Brake_{c}(\speed{c})  = \distance{}(\po{c},\limit{c}) \\
    & \Longleftrightarrow &
    (\exists \limit{c})_c~ \phi_\pi \wedge \bigwedge\nolimits_c
    \Brake_{c}(\speed{c}) = \distance{}(\po{c},\limit{c})
    \end{array} \]
  (ii) Immediate consequence of (i) by applying $\swarrow$ on both sides.
\end{proof}

Re-calling Thm.~\ref{thm:safe-control-policy} in section \ref{subsec:ag-contracts}, notice that
$\Brake_{c}(\speed{c}) \le \distance{}(\po{c}, \limit{c})$ is enforced by
safe control policies as $\distance{}(\po{c}, \limit{c})=\fs{c}$.
Therefore, any property $\phi$ is preserved through equivalence only when
all the vehicles run with the maximal allowed speed by the distance to
their limit positions. Otherwise, the speed-lower closure $\cl{\phi}$ is
preserved through equivalence, that is, only the upper bounds on speeds
as derived from corresponding bounds on limit positions.


Therefore, all traffic rules of form
(\ref{eq:traffic-rule}) which, for states satisfying the precondition $\phi$,
constrain the speed of vehicle $c_1$ at the next cycle according to
constraint $\psi$, are transformed into free space rules on limit positions
 of the form:
\begin{equation}
  \label{eq:free-space-rule}
  \always ~ \forall c_1. \forall o_2...\forall o_k.~ \phi_\pi(c_1,o_2,\dots,o_k) \implies \next \psi_\pi(c_1,o_2,\dots,o_k)
\end{equation}
Notice that the postcondition $\psi_\pi$ is of the form
$\limit{c_1} \le_{c_1} b_\psi(c_1,o_2,\dots,o_k)$ for a position term $b_\psi$ obtained by the transformation of $\psi$.

For example, the traffic rule $1$ is transformed into the free space rule: 
$\always \forall c_1.~\forall c_2.$ $\meets{c_1}{c_2}$ $\implies$ $\next \limit{c_1} \le_{c_1} \po{c_2}$. 
The traffic rule $4(ii)$ is transformed into the free space rule: 
$\always \forall c_1.~\forall c_2.~\forall u.~ \edgemeets{c_1}{u} \wedge \edgemeets{c_2}{u} \wedge
\limit{c_1} = u \wedge \limit{c_2} = u \wedge \po{c_1} \priority{u} \po{c_2} \implies \next \limit{c_1} \le_{c_1} u$.

We are now ready to define the \Runtime\ free space policy based on traffic rules.
Let ${\cal R}$ denotes the set of traffic rules of interest e.g., the ones defined in Table~\ref{tab:traffic-rules}.  
For a current \ADS\ state $st$ and current limit positions and free spaces $\langle
\limit{c}, \fs{c} \rangle_{c\in\Vehicles}$ the policy computes
new limit positions and new free spaces $\langle \limit{c}', \fs{c}' \rangle_{c\in\Vehicles}$ as follows:
\begin{eqnarray}
  \limit{c}' & \eqdef &  \min_{\le c} ~\{~ \sigma b_\psi ~|~
   [\forall c_1. \forall o_2...\forall o_k.~ \phi \implies N \psi] \in {\cal R},~~
   \sigma[c/c_1],st \vdash \phi_\pi \} \nonumber \\
   & & \hspace{0.8cm} \cup ~ \{ ~ \post{e} ~|~ \exists a < \segnorm{e},~ \limit{c} = (e,a),~\meets{c}{e} ~ \} \\
  \fs{c}' & \eqdef & \delta \mbox{ such that } \po{c} \oplus_c \delta = \limit{c}'
\end{eqnarray}
Actually, that means computing for every vehicle $c$ the new limit position
$\limit{c}'$ as the nearest position with respect to $\le_{c}$ from two
sets of bounds.  The first set contains the bounds $\sigma b_\psi$ computed
for all the free space rules derived from the traffic rules in ${\cal R}$
and applicable for $c$ at the given state $st$. The second set contains the
endpoint $\post{e}$ of the edge $e$ where the current limit position
$\limit{c}$ is located.  It is needed to avoid "jumping" over $\post{e}$,
even though this is allowed by application of the rules, as
$\post{e}$ may be a merger node and should be considered for solving
potential conflicts.  Then, we define the new free space $\fs{c}'$ as the
distance $\delta$ from the current position $\po{c}$ to the new limit
position $\limit{c}'$ measured along the itinerary of $c$.

Note that if the free space policy respects the assume-guarantee contract
of the \Runtime\ from Def.~\ref{def:safe-control-policy} then it will
moreover guarantees the satisfaction of all traffic rules from ${\cal R}$
where both the pre- and the postcondition $\phi$ and $\psi$ are speed-lower
closed formulas.  First, conformance with respect to the contract is needed
to obtain the invariants $\Brake_c(\speed{c}) \le \fs{c} =
\distance{}(\po{c},\limit{c})$ according to
Thm.~\ref{thm:safe-control-policy}.  Second, these invariants ensure
preservation through equivalence between speed-lower closed formula and
derived formula on limit positions, according to
Thm.~\ref{thm:preservation}.  Third, the free space policy ensures the
satisfaction of the derived free space rules, that is, by construction it
chooses limit positions ensuring postconditions $\psi_\pi$ hold whenever
preconditions $\phi_\pi$ hold.  As these formulas are preserved through
equivalence, it leads to the satisfaction of the original traffic rule.

\subsection{Correctness with respect to the Free Space Contract}
We prove correctness, that is, conformance with the assume-guarantee
contract of Def.~\ref{def:safe-control-policy}, of the free space policy
obtained by the application of the traffic rules from
Table.~\ref{tab:traffic-rules} excluding the one concerning traffic lights.
For this rule we need additional assumptions taking into account the light
functioning and the behavior of the crossing vehicles.

First, we
assume that the vehicle braking dynamics are compatible with the speed
limits associated with the map segments, that means:
\begin{itemize}
\item for any edge $e$ leading to a junction (and henceforth a stop sign)
  or a merger vertex holds $\Brake_c(\speed{e}) \le \segnorm{e}$, for
  any vehicle $c\in \Vehicles$ (see Figure~\ref{fig:speed-limits}(a)),
\item for any consecutive edges $e_1$, $e_2$ holds
  $\Brake_{c}(\speed{e_1}) \le \segnorm{e_1} + \Brake_c(\speed{e_2})$, for
  any vehicle $c\in \Vehicles$ (see Figure~\ref{fig:speed-limits}(b)) i.e.,
  between two consecutive speed limit changes, there is sufficient space to
  adapt the speed.
\end{itemize}

  \begin{figure}[htbp]
    \vspace{-3mm}
    \begin{center}
      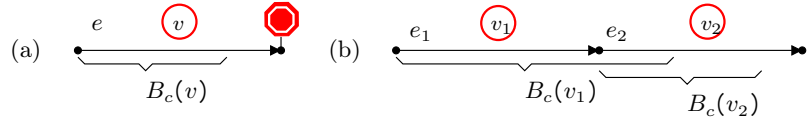
    \end{center}
    \vspace{-3mm}
    \caption{\label{fig:speed-limits}{Explaining restrictions on speed limits}}
    \vspace{-3mm}
  \end{figure}

We call an \ADS\ state $\langle st_o \rangle_{o\in \Objects}$
\emph{consistent with limit positions} $\langle \limit{c} \rangle_{c\in\Vehicles}$ iff
for every vehicle $c \in \Vehicles$:
\begin{itemize}
\item the limit position is ahead of the current vehicle
  position, that is, $\po{c} \le_c \limit{c}$,
\item there is no stop sign located strictly between the current vehicle position and the
  limit position, that is, $\po{c} <_c \po{so} <_c \limit{c}$ does not hold for any stop $so$,
\item the limit position conforms to the speed limits of the current edge
  ($e_1$) and next edge ($e_2$) on the itinerary of $c$, that is, $\distance{}(\po{c},\limit{c})
  \le \Brake_c(\speed{e_1})$ and $\distance{}(\po{c},\limit{c}) \le
  \distance{}(\po{c}, \pre{e_2}) + B_c(\speed{e_2})$.
\end{itemize}
For vehicle $c$ and position $p$ located ahead of $c$ on its itinerary we
further denote by $\ahead{c}{p} \eqdef \Ahead{c}{f}$ for $f=d(\po{c},p)$,
that is, the space ahead of $c$ until position $p$.  In particular,
$\ahead{c}{\limit{c}} = \Ahead{c}{\fs{c}}$ holds as $\limit{c} = \po{c}
\oplus_c \fs{c}$.  The following lemma provides the basic conditions that
guarantee the correctness of the free space policy.
\begin{lemma}\label{lem:consistency}
  Let $st$ be an \ADS\ state and
  $\langle \limit{c} \rangle_{c\in\Vehicles}$ be limit positions such that:
  \begin{itemize}
  \item the state $st$ is consistent with the limit positions $\langle \limit{c} \rangle_{c\in\Vehicles}$,
  \item the spaces ahead up to the limit positions are non-crossing,
    that is, $\biguplus_{c\in\Vehicles} \ahead{c}{\limit{c}}$.
  \end{itemize}
  Let $\langle \limit{c}' \rangle_{c\in\Vehicles}$ be the new limit positions computed for
  state $st$ and $\langle \limit{c} \rangle_{c\in\Vehicles}$ according to the free
  space policy.  Then, the following hold:
  \begin{enumerate}[label=(\alph*)]
  \item $\limit{c} \le_c \limit{c}'$, for every vehicle $c\in \Vehicles$,
  \item the state $st$ is consistent with the new limit positions $\langle \limit{c}' \rangle_{c\in\Vehicles}$,
  \item the spaces ahead up to the new limit positions are non-crossing,
    that is, $\biguplus_{c\in\Vehicles} \ahead{c}{\limit{c}'}$.
  \end{enumerate}
\end{lemma}
\begin{proof}
  (a) Let fix an arbitrary vehicle $c$.  According to the free space policy
  the limit position $\limit{c}'$ is defined as the nearest position
  among several bounds on the itinerary of $c$.  We will show that, in
  all situations, these bounds $b$ are at least as far as the current limit
  position $\limit{c}$, that is $\limit{c} \le_c b$ and hence the result.
  First, consider all applicable traffic rules and their associated bounds:
  \begin{itemize}
  \item rule $1$: The bound $b$ is defined as the position $\po{c_2}$ of the
    heading vehicle $c_2$.  The constraint $\limit{c} \le_c \po{c_2}$ holds because
    the spaces ahead to the limit positions are assumed non-crossing,
  \item rule $2(i)$: The bound $b$ is defined as the position
    $\po{so}$ of the stop sign located ahead of $c$.  The constraint
    $\limit{c} \le_c \po{so}$ holds because the state $st$ is assumed
    consistent, that is, no stop sign between the vehicle and the
    current limit position,
  \item rules $2(ii,iii,iv)$: The bound $b$ is defined as the current
    position $\po{c}$, which is equal to the current limit $\limit{c}$ due
    to the constraint $\speed{c} = 0$ in their respective preconditions,
    therefore obviously $\limit{c} \le_c \limit{c}$,
  \item rules $4(i,ii)$: The bound $b$ is defined as the merger node $u$ and the
    preconditions of the rules contain respectively $\limit{c}
    \le_c u$ and $\limit{c} = u$,
  \item rules $5(i,ii)$: Assume $c$ is located on some edge $e_1$ and the next
    edge is $e_2$.  As $st$ is consistent with current limit positions,
    this implies $\distance{}(\po{c}, \limit{c}) \le \Brake(\speed{e_1})$
    and $\distance{}(\po{c}, \limit{c}) \le \distance{}(\po{c}, \pre{e_2})
    + \Brake_c(\speed{e_2})$.  This is equivalent to $\limit{c} \le_c
    \po{c} \oplus_c \Brake_c(\speed{e_1})$ and $\limit{c} \le_c \pre{e_2}
    \oplus_c \Brake_c(\speed{e_2})$ which are the constraints for speed limit rules.
  \end{itemize}
  Second, we consider the position $\post{e}$, where the edge $e$
  contains the limit position $\limit{c}$.  That is, $\limit{c} = (e,a)$
  for some $a < \segnorm{e}$ and hence $\limit{c} \le_c (e,\segnorm{e}) =
  \post{e}$.

  (b) We know that $\po{c} \le_c \limit{c}$ for all vehicles $c$.  Then,
  from (a) above we obtain immediately that $\po{c} \le_c \limit{c}'$ for
  all vehicles $c$.  Moreover, according to traffic rule $2(i)$, new limit
  positions $\limit{c}'$ cannot move over a stop sign $so$ unless the
  vehicle $c$ is at $so$.  That means, stop signs cannot be inserted
  between a vehicle and their limit positions.  Finally, according to
  traffic rules $5(i,ii)$, the new limit positions are at most as far as
  the bounds for the current and next edge speed limits.  The restriction
  on speed limits from Figure~\ref{fig:speed-limits}(b) is needed whenever
  vehicles are located at vertices, because of the transition from the
  current to the next edge.

  (c) We have seen that limit positions can either stay unchanged or
  move forward on the itineraries of their respective vehicles.  In
  order to show that spaces until the new limit positions are
  non-crossing, we need to focus on moving limit positions.  We prove
  the following facts, which guarantee that these spaces remain
  indeed non-crossing:
  \begin{itemize}
  \item \emph{no limit position exceeds the current position of a
  vehicle}. Actually, this is explicitly excluded by the traffic rule
    $1$.  Therefore, no space is extended so as to overlap
    with an existing space,
  \item \emph{two limit positions never move simultaneously so that the corresponding
    spaces cross each other}.  First, crossing may happen at
    junctions i.e., if two limit positions will simultaneously enter
    the same junction.  Such situations are excluded by traffic rules
    $2(i,ii,iii,iv)$ which force vehicles to stop and
    then solve conflicts between them in a
    deterministic manner.  Second, crossing may happen at merger nodes
    i.e., if two limit positions will simultaneously move over a
    merger node.  These situations are explicitly excluded by traffic
    rule $4(i,ii)$ which solve conflicts at merger nodes based on
    priorities, plus the extra rule forbidding limit positions to jump
    over map vertices. \qed
  \end{itemize}
\end{proof}

\begin{proposition}
  The free space policy respects the safety contract for the
  \Runtime\ provided the initial \ADS\ state is consistent with initial
  limit positions.
\end{proposition}
\begin{proof}
  Let consider a state $st$ consistent with limit positions $\langle
  \limit{c} \rangle_{c\in\Vehicles}$ and satisfying the assumptions stated
  in Def.~\ref{def:safe-control-policy}. We are therefore satisfying the
  premises of Lemma~\ref{lem:consistency} and consequently the new limit
  positions $\langle \limit{c}' \rangle_{c\in \Vehicles}$ will satisfy the conclusions
  (a)-(c) as stated by the Lemma~\ref{lem:consistency}.

  Then, for any vehicle $c$, using (a) we know that the new limit position
  satisfies $\limit{c} \le_c \limit{c}'$.  Consequently, the new free space
  $\fs{c}'$ which is the
  distance from the current vehicle position to the next limit position
  satisfies $\fs{c}' \ge \fs{c} - \dt{c}$, that is, the first assertion of
  the \Runtime\ guarantee in Def.~\ref{def:safe-control-policy}.  This inequality can be
  understood from Fig.~\ref{fig:update-limit-free-space} which depicts the
  generic situation for a vehicle $c$.

  \begin{figure}[htbp]
    \vspace{-3mm}
    \begin{center}
      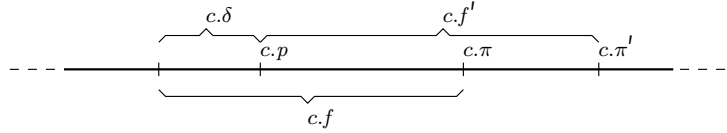
    \end{center}
    \vspace{-3mm}
    \caption{\label{fig:update-limit-free-space}{Update of limit position and
    free space along the itinerary of $c$}}
    \vspace{-3mm}
  \end{figure}

  Moreover, using (c) we obtain immediately the second assertion of the
  \Runtime\ guarantee, that is, $\biguplus_{c\in \Vehicles}
  \Ahead{c}{\fs{c}'}$.

  Finally, using (b) we know that the state $st$ is consistent with new
  limit positions $\langle \limit{c}' \rangle_{c\in\Vehicles}$.  Then, as long as the
  vehicles move forward into their free spaces according to their contract, the system state
  $st'$ is consistent with respect to these new limit positions.  So,
  eventually, at the beginning of the next \Runtime\ cycle we are back to
  the initial situation considered for $st$ and limit positions $\langle
  \limit{c} \rangle_{c\in\Vehicles}$.  That means essentially that state
  consistency with respect to limit positions is an inductive invariant in
  the system, so it can be safely assumed to hold at any time provided it
  holds initially. \qed
\end{proof}

\subsection{Non-blocking Free Space Policies}\label{subsec:freespace-non-blocking}

The free space policy based on traffic rules is not non-blocking.  While
continuously meeting traffic rules, an \ADS\ can potentially evolve into a
blocking state.  For example, when a subset of vehicles is blocked in a roundabout so that none of them can move further and eventually exit the roundabout and all other
vehicles are waiting to enter the same roundabout.  The non-blocking
contract can be however fulfilled if the \Runtime\ monitors the traffic
from a global point of view and prevents situations as described above. 

Let consider that there exists a constant $f_{min}$ such that for all map
edges $e$, $\segnorm{e} > f_{min}$ and $\speed{e} > \Brake^{-1}(f_{min})$.
Furthermore, let $E_\crossing \subseteq E$ be the subset of edges belonging
to junctions.  An elementary directed path $\gamma = e_1 e_2 \dots e_m \in
E^*$ is \emph{critical} either if (i) it is a circuit visiting at most once
every junction or (ii) it ends and returns at the same junction, while
visiting at most once every other junction.  Let $\#_j(\gamma)$ be the number
of distinct junctions of a critical path $\gamma$.  We define the capacity
$\capacity{\gamma}$ of a critical path $\gamma$ as the least number of vehicles
that could "block" the critical path $\gamma$ minus one:
\begin{equation}
   \capacity{\gamma} \eqdef \#_j(\gamma) + \left(\sum\nolimits_{e \in \gamma \setminus E_\crossing}
   \big\lfloor \segnorm{e} ~/~ f_{min} \big\rfloor\right)  - 1
\end{equation}
That is, we assume that a junction can be blocked by one vehicle, and a
non-junction edge $e$ can be blocked by $\lfloor \segnorm{e} ~/~f_{min}
\rfloor$ vehicles.
\begin{proposition}
  The free space policy derived from traffic rules respects the
  non-blocking contract for the \Runtime\ as long as the number of vehicles
  on every critical path $\gamma$ is lower than the path capacity $\capacity{\gamma}$.
\end{proposition}
\begin{proof}
  Consider a state where all vehicles are stopped.  Consider an arbitrary
  vehicle $c_0$.  Two situations can happen, respectively (i) $c_0$ is
  waiting at an entry of a junction since another vehicle is already in or
  (ii) $c_0$ is waiting behind another stopped vehicle.  That means, in
  both cases, some other vehicle $c_1$ is actually blocking $c_0$.  We
  continue the same reasoning for $c_1$, and can find another vehicle $c_2$
  blocking it and so on.  As the number of vehicles is finite, we finally find a circular
  chain  $c_i$, $c_{i+1}$, $\dots$, $c_n$, $c_i$ of vehicles that block each other successively.

  Any chain of vehicles as above is located on some critical path $\gamma$,
  that is, either a circuit or a path going twice through the same junction
  of the map.  Assume the number of vehicles on the path $\gamma$ is lower than
  the path capacity $\capacity{\gamma}$.  We distinguish two situations:
  \begin{itemize}
  \item \emph{the critical path contains only junction edges:} Then, as the
    path capacity is $\#_j(\gamma) - 1$ at least one of the junctions must be
    clear of vehicles.  Therefore, at least one of the vehicles waiting at
    that junction entries will obtain the free space to proceed.  That is,
    the vehicle in the preceding junction or another one waiting at
    some other entry, depending on the applicable traffic rules.
  \item \emph{the critical path contains both junction and non-junction
    edges:} If at least one of the junctions is not empty, we can reason as
    in the previous case and find a vehicle that can proceed.  Otherwise,
    as every junction contains one vehicle, then the number of vehicles on
    the remaining non-junction edges is at most $\sum\nolimits_{e \in
      \gamma \setminus E_\crossing} \big\lfloor \segnorm{e} /f_{min} \big\rfloor - 1$. That
    means, one can find a non-junction edge with more than $f_{min}$
    unused space.  That space is eventually allocated to either one of the
    vehicles waiting on the edge, or to a vehicle waiting to enter the
    edge (that is, exiting from a junction, entering through a merger node,
    etc), depending on applicable rules.\qed
  \end{itemize}
\end{proof}


\section{Experiments}\label{sec:experiments}

We are currently developing a prototype for \ADS\ simulation implementing
the speed and free space policies introduced.  It takes as inputs (i) a map
defined as an annotated metric graph with segments that are parametric
lines or arcs and (ii) an initial state containing the positions, the
initial speeds and the itineraries of a number of vehicles.  The simulation
proceeds then as explained in section~\ref{sec:dynamic-model} by using the
specific control policies from sections~\ref{sec:speed-policies}
and \ref{sec:freespace-policies}.  The prototype uses the
SFML\footnote{Simple and Fast Multimedia
Library, \url{https://www.sfml-dev.org/}} library for graphical rendering
of states.  Fig.~\ref{fig:experiments} presents simulation snapshots for
two simple examples running respectively 5 and 18 vehicles. Note that
performance scales up smoothly as the number of vehicles increases because
each rule is applied on a linear finite horizon structure. Furthermore,
compared to simulators that particularize an \emph{ego} vehicle,
the \Runtime\ treats all vehicles the same and ignores their speed control
policy as long as they fulfill their contract.

\begin{figure}[thbp]
\centering
\begin{subfigure}[c]{.5\linewidth}
  \centering
  \includegraphics[width=.9\linewidth]{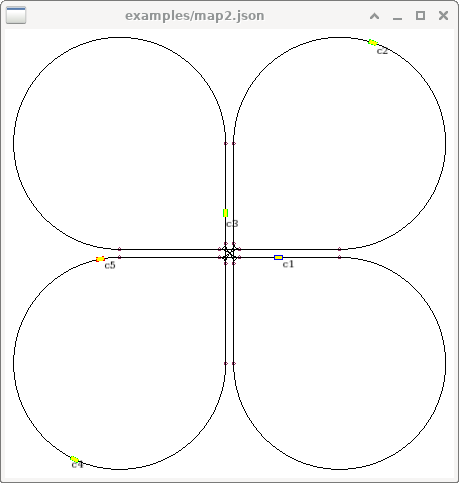}
  \label{fig:sub1}
\end{subfigure}%
\begin{subfigure}[c]{.5\linewidth}
  \centering
  \includegraphics[width=.9\linewidth]{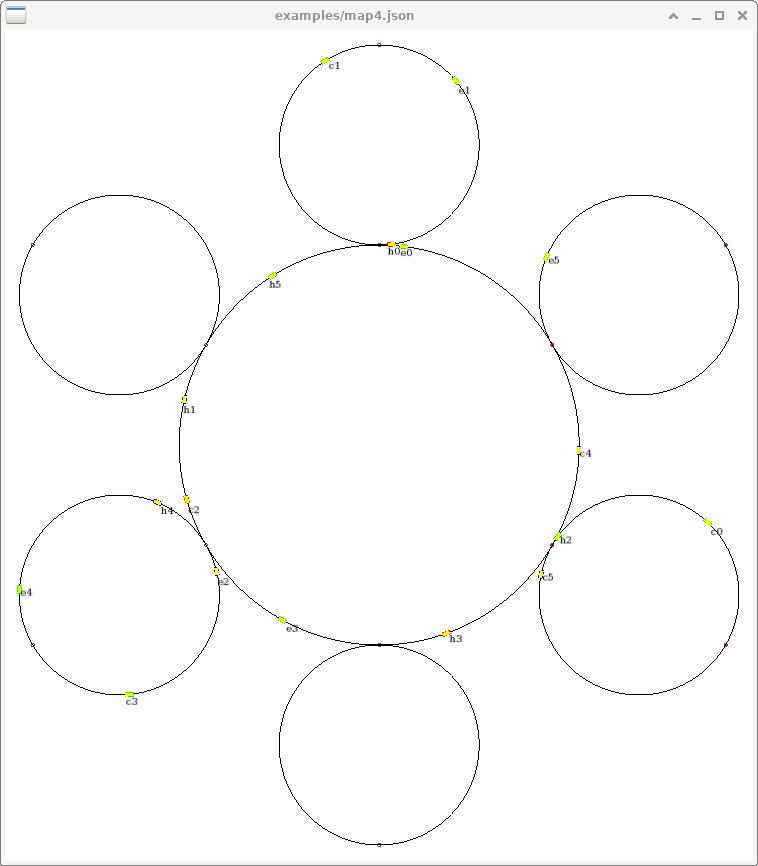}
  \label{fig:sub2}
\end{subfigure}
\caption{Snapshots of \ADS\ simulation}
\label{fig:experiments}
\vspace{-5mm}
\end{figure}


\section{Discussion}\label{sec:discussion}
The paper studies results for the correct by design coordination of
\ADS\ based on assume-guarantee contacts.  The coordination follows a
two-step synchronous interaction protocol between vehicles and a
\Runtime\ that, based on the distribution of vehicles on a map, computes
the corresponding free spaces. A first result characterizes safe
control policies as the combination of assume-guarantee contracts for
vehicles and the \Runtime. This result is then specialized by showing
how policies consistent with their respective contracts can be defined
for vehicles and the \Runtime. In particular, for vehicles, we provide
a principle for defining speed policies and, for the \Runtime, we
compute free space policies that conform to a set of traffic rules.
The results are general and overcome the limitations of
a posteriori verification. They can be applied to \ADS\ involving a
dynamically changing number of vehicles. In addition, they rely on a general
map-based environment model, which has been extensively studied in
\cite{abs-2109-06478}. Control policies for vehicles and the runtime can be
implemented efficiently. In particular, the speed policy has been
tested in various implementations \cite{WangLS20,abs-2103-15484} and
found to be not only safe, but also closer to the optimum when
refining the space of possible accelerations.

Note that the results can be
extended with slight modifications to maps where the segments are
curves or regions to express traffic rules involving properties of
two-dimensional space, for example for passing maneuvers.  For
example, if we consider region maps, their segments will be regions of
constant width centered on curves. Itineraries, free spaces and $B(v)$
will be regions. The relationship $B(v)\le f$ becomes $B(v)\subseteq
f$ and the addition of lengths of segments should be replaced by the
disjoint union of the regions they represent. The speed control policy
will remain unchanged in principle but will need a function computing
the distance travelled in a region. Finally, the runtime verification
of the disjointness of free spaces may incur a computational cost
depending on the accuracy of the region representation.

The presented results provide a basis for promising developments in several
directions. One direction is to extend the results to achieve correctness
by design for general properties.  We have shown that traffic rules, which
are declarative properties of vehicles, can be abstracted into safety
constraints on free spaces.  In this way, we solved a simple synthesis
problem by transforming a “static” constraint on vehicle speed into a
“dynamic” constraint on shared resources.

An interesting question that should be further investigated, is
whether the method can be extended to more general properties
involving the joint obligation of many vehicles. For example, we can
require that for any pair of vehicles $c_1$ and $c_2$ that are
sufficiently close, the absolute value of the difference between their
speeds is less than a constant $k$, i.e., $|\speed{c_1}-\speed{c_2}|
\le k$. This can be achieved by a free space constraint that gives
more free space to the vehicle with the lower speed, assuming that
vehicle speed policies are not "lazy" and use as soon as possible the
available space.

For general properties involving more than one vehicle, it seems realistic
to translate them directly into free space constraints that will enforce
the constraints processed by the \Runtime\ to ensure the safe control
policy. In particular, in addition to safety properties, we could devise
free space policies that optimize criteria such as road occupancy and
uniform separation for a given group of vehicles e.g. platoon systems
studied in \cite{El-HokayemBBS20}.  Note that achieving non-blocking
control is such a property that involves the application of occupancy
criteria. 

Another direction is to move from centralized to
distributed coordination with many runtimes. It seems possible to
partition traffic rules according to the geometric scope of their
application, e.g., a specific runtime could control access to each
junction.  Finally, the \Runtime\ can be used as a monitor to
verify that the vehicle speed policies of an \ADS\ are safe and
respect the given traffic rules.

\bibliographystyle{splncs04}
\bibliography{biblio-compact}

\end{document}